\documentclass[english]{article}
\usepackage[T1]{fontenc}
\usepackage[latin9]{inputenc}
\usepackage{babel}
\usepackage{float}
\usepackage{enumitem}
\usepackage{amsmath}
\usepackage{amsthm}
\usepackage[unicode=true,pdfusetitle,
 bookmarks=true,bookmarksnumbered=false,bookmarksopen=false,
 breaklinks=false,pdfborder={0 0 0},pdfborderstyle={},backref=false,colorlinks=false]
 {hyperref}

\makeatletter

\floatstyle{ruled}
\newfloat{algorithm}{tbp}{loa}
\providecommand{\algorithmname}{Algorithm}
\floatname{algorithm}{\protect\algorithmname}

\numberwithin{equation}{section}
\numberwithin{figure}{section}
  \theoremstyle{plain}
  \newtheorem{lem}{\protect\lemmaname}[section]
 \ifx\proof\undefined\
   \newenvironment{proof}[1][\proofname]{\par
     \normalfont\topsep6\p@\@plus6\p@\relax
     \trivlist
     \itemindent\parindent
     \item[\hskip\labelsep
           \scshape
       #1]\ignorespaces
   }{%
     \endtrivlist\@endpefalse
   }
   \providecommand{\proofname}{Proof}
 \fi
 \newlist{casenv}{enumerate}{4}
 \setlist[casenv]{leftmargin=*,align=left,widest={iiii}}
 \setlist[casenv,1]{label={{\itshape\ \casename} \arabic*.},ref=\arabic*}
 \setlist[casenv,2]{label={{\itshape\ \casename} \roman*.},ref=\roman*}
 \setlist[casenv,3]{label={{\itshape\ \casename\ \alph*.}},ref=\alph*}
 \setlist[casenv,4]{label={{\itshape\ \casename} \arabic*.},ref=\arabic*}
  \theoremstyle{plain}
  \newtheorem{cor}{\protect\corollaryname}[section]
  \theoremstyle{remark}
  \newtheorem{claim}{\protect\claimname}[section]

\usepackage{algpseudocode}
\numberwithin{algorithm}{section}
\algrenewcommand\textproc{\texttt}

\makeatother

 \providecommand{\casename}{Case}
\providecommand{\claimname}{Claim}
\providecommand{\corollaryname}{Corollary}
\providecommand{\lemmaname}{Lemma}

\begin{document}

\title{The Life in 1-Consensus}

\author{Yehuda Afek\\
School of Computer Science\\
Tel-Aviv University\\
\texttt{<afek@post.tau.ac.il>}
\and Eli Daian\\
School of Computer Science\\
Tel-Aviv University\\
\texttt{<eliyahud@post.tau.ac.il>}
\and Eli Gafni\\
Department of Computer Science\\
University of California, Los-Angeles\\
\texttt{<eli@cs.ucla.edu>} } \maketitle
\begin{abstract}
This paper introduces the atomic Write and Read Next ($\text{WRN}_{k}$)
deterministic shared memory object, that for any $k\ge3$, is stronger
than read-write registers, but is unable to implement $2$-processor
consensus. In particular, it refutes the conjecture claiming that
every deterministic object of consensus number $1$ is computationally
equivalent to read-write registers.
\end{abstract}

\section{Introduction}

Shared memory objects have been classified by Herlihy \cite{H91}
by their \emph{consensus number}, the number of processes which can
solve wait-free consensus using any number of copies of an object
and atomic read-write registers.

At the same time, it was also proved that $n$-consensus objects are
\emph{universal} for systems of $n$ processes. However, the computational
power of objects with consensus number $n$ in systems of more than
$n$ processes is not completely understood. Recently, \cite{Afek:2016:DOL:2933057.2933116}
have constructed an infinite sequence of deterministic Set Consensus
objects of consensus number $n\ge2$, with strictly increasing computational
power in systems of more than $n$ processes. However, the case for
$n=1$ remained an open question. And, it has been conjectured that
any deterministic object of consensus number $1$ is computationally
equivalent to read-write registers, meaning that it cannot solve any
problem that is not solvable by read-write registers.

For the nondeterministic case, \cite{Herlihy:1991:IRA:113379.113409}
showed a counter example: a non-deterministic object with consensus
number 1 that cannot be implemented from read-write registers. Deterministic
Set Consensus objects, similar to the ones used in \cite{Afek:2016:DOL:2933057.2933116},
do not provide such a hierarchy for the case of $n=1$, because they
can be used to solve the consensus task for two processes. It is done
by inspecting them as deterministic state machines, and initializing
in a way that it is possible to predict the decided values for subsequent
processors.

In this note, we refute the above conjecture by constructing a deterministic
object, Write and Read Next ($\text{WRN}_{k}$), that solves $\left(k,k-1\right)$-set
consensus, but cannot solve $2$ processors consensus, for every $k\ge3$.
We define the $\text{WRN}_{k}$ object in section \ref{sec:Set-and-Fetch}.
We show that these objects cannot solve $2$-consensus for any $k\ge3$
in section \ref{sec:-is-Weaker}, and we show that they can be used
to implement set consensus in section \ref{sec:-Solves-Set}. The
$k\le2$ case is briefly discussed in section \ref{sec:The--Case}.

\section{The Model}

We follow the standard asynchronous shared memory model, as defined
in \cite{Afek:2016:DOL:2933057.2933116}, in which processes communicate
with one another by applying atomic operations, called steps, to shared
objects. Each object has a set of possible values or states. Each
operation (together with its inputs) is a partial mapping, taking
each state to a set of states. A shared object is \emph{deterministic}
if each operation takes each state to a single state and its associated
response is a function of the state to which the operation is applied.

A \emph{configuration} specifies the state of every process and the
value of every shared object. An \emph{execution} is an alternating
sequence of configurations and steps, starting from an initial configuration.
A faulty process can stop taking steps, but, otherwise, must behave
in accordance with the algorithm it is executing. If $C$ is a configuration
and $s$ is a sequence of steps, we denote by $Cs$ the configuration
(or in the case of nondeterministic objects, the set of possible configurations)
when the sequence of steps $s$ is performed starting from configuration
$C$.

An \emph{implementation} of a sequentially specified object $O$ consists
of a representation of $O$ from a set of shared base objects and
algorithms for each process to apply each operation supported by $O$.
The implementation is \emph{deterministic} if all its algorithms are
deterministic. The implementation is \emph{linearizable} if, in every
execution, there is a sequential ordering of all completed operations
on $O$ and a (possibly empty) subset of the uncompleted operations
on $O$ such that:
\begin{enumerate}
\item If $op$ is completed before $op^{\prime}$ begins, then $op$ occurs
before $op^{\prime}$ in this ordering.
\item The behavior of each operation in the sequence is consistent with
its sequential specification (in terms of its response and its effect
on shared objects).
\end{enumerate}
An implementation of an object $O$ is \emph{wait-free} if, in every
execution, each process that does not crash completes each of its
operations on $O$ in a finite number of its own steps. The implementation
is \emph{non-blocking} if, starting from every configuration in every
infinite execution, some process completes one of its operations within
a finite number of steps. In the rest of this paper, we discuss only
deterministic, linearizable and wait-free implementations.

A \emph{task} specifies what combinations of output values are allowed
to be produced, given the input value of each process and the set
of processes producing output values. A wait-free or non-blocking
solution to a task is an algorithm in which each process that does
not crash produces an output value in a finite number of its own steps
such that the collection of output values satisfies the specification
of the task, given the input values of the process.

In the \emph{consensus task}, each process, $p_{i}$, has an input
value $x_{i}$ and, if it is non-faulty, must output a value $y_{i}$
that satisfies the following two properties:
\begin{description}
\item [{Validity}] Every output is the input of some process.
\item [{Agreement}] All outputs are the same.
\end{description}
We say that an execution of an algorithm solving consensus \emph{decides
a value} if that value is the output of some process. \emph{Binary
consensus} is the restriction of the consensus task in which each
input value $x_{i}\in\left\{ 0,1\right\} $.

The \emph{$k$-set consensus task}, introduced by \cite{Chaudhuri:1990:AHC:93385.93431,C93},
is defined in the same way, except that agreement is replaced by the
following property:
\begin{description}
\item [{$k$-agreement}] There are at most $k$ different output values.
\end{description}
Note that the $1$-set consensus task is the same as the consensus
task.

An object has \emph{consensus number $n$} if there is a wait-free
algorithm that uses only copies of this object and registers to solve
consensus for $n$ processes, but there is no such algorithm for $n+1$
processes. An object has an infinite consensus number if there is
such algorithm for each positive integer $n$.

For all positive integers $k<n$, an $\left(n,k\right)$-set consensus
nondeterministic object \cite{borowsky1993implication} supports one
operation, \texttt{propose}, which takes a single non-negative integer
as input. The value of an $\left(n,k\right)$-set consensus object
is a set of at most $k$ values, which is initially empty, and a count
of the number of \texttt{propose} operations that have been performed
on it (to a maximum of $n$). The first \texttt{propose} operation
adds its input to the set. Any other \texttt{propose} operation can
nondeterministically choose to add its input to the set, provided
the set has size less than $k$. Each of the first $n$ \texttt{propose}
operations performed on the object \emph{nondeterministically} returns
an element from the set as its output. All subsequent \texttt{propose}
operations return $\bot$.

\section{\label{sec:Set-and-Fetch}Write and Read Next Objects}

For every $k\ge2$, we introduce the $\text{WriteAndReadNext}_{k}$
(or $\text{WRN}_{k}$) object, that has a single operation \textendash{}
\texttt{WRN}. This operation accepts an index $i$ in the range $\left\{ 0,\dots,k-1\right\} $,
and a value $v\neq\bot$. It returns the value $v^{\prime}$ that
was passed in the previous invocation to \texttt{WRN} with the index
$\left(i+1\right)\mod k$, or $\bot$ if there is no such previous
invocation.

A possible implementation of $\text{WRN}_{k}$ consists of $k$ registers,
$A\left[0\right],\dots,A\left[k-1\right]$, initially initialized
to $\bot$. A sequential specification of the atomic \texttt{WRN}
operation is presented in algorithm \ref{alg:A-sequential-specification}.

\begin{algorithm}
\begin{algorithmic}[1]
\Function{WRN}{$i, v$}
    \Comment {$i \in \left\{ 0, \dots, k - 1 \right\}$, $v \neq \bot$}
  \State $A \left[ i \right] \gets v$
  \State \Return{$A \left[ \left( i + 1 \right) \mod k \right]$}
\EndFunction
\end{algorithmic}

\caption{\label{alg:A-sequential-specification}A sequential specification
of the atomic \texttt{WRN} operation of a $\text{WRN}_{k}$ object.}
\end{algorithm}

\medskip{}

From now on, we assume $k\ge3$, unless stated otherwise.

\section{\label{sec:-is-Weaker}$\text{WRN}_{k}$ is Weaker than $2$-Consensus}

We follow the standard definitions of \emph{bivalent configuration},
\emph{$v$-univalent configuration} and \emph{critical configuration},
as defined in \cite{FLP,H91}.
\begin{lem}
\label{lem:For-each-,}For each $k\ge3$, there is no wait-free algorithm
for solving the consensus task with $2$ processes using only registers
and $\text{WRN}_{k}$ objects.
\end{lem}
\begin{proof}
Assume such an algorithm exists. Consider the possible executions
of the processes $P$ and $Q$ of this algorithm, while proposing
$0$ and $1$, respectively. Let $C$ be a critical configuration
of this run. Denote the next steps of $P$ and $Q$ from $C$ as $s_{P}$
and $s_{Q}$, respectively. Without loss of generality, we assume
that $Cs_{P}$ is a $0$-univalent configuration, and $Cs_{Q}$ is
a $1$-univalent configuration.

Following \cite{H91}, $s_{P}$ and $s_{P}$ both invoke a \texttt{WRN}
operation on the same $\text{WRN}_{k}$.
\begin{casenv}
\item Both $s_{P}$ and $s_{Q}$ perform \texttt{WRN} with the same index
$i$.

The configurations $Cs_{P}$ and $Cs_{Q}s_{P}$ are indistinguishable
for a solo run of $P$, but a solo run of $P$ from $Cs_{P}$ decides
$0$, while an identical solo run of $P$ from $Cs_{Q}s_{P}$ decides
$1$. This is a contradiction.
\item $s_{P}$ and $s_{Q}$ perform \texttt{WRN} with different indices,
$i_{P}$ and $i_{Q}$, respectively.

Since $k\ge3$, either $i_{P}\neq i_{Q}+1\mod k$ or $i_{Q}\neq i_{P}+1\mod k$.
Without loss of generality, assume that $i_{Q}\neq i_{P}+1\mod k$.
So the configurations $Cs_{P}s_{Q}$ and $Cs_{Q}s_{P}$ are indistinguishable
for a solo run of $P$. However, the identical solo runs of $P$ from
the configurations $Cs_{P}s_{Q}$ and $Cs_{Q}s_{P}$ decide $0$ and
$1$, respectively, which is a contradiction.

\end{casenv}
Both cases resulted in a contradiction, and therefore no such algorithm
exists.
\end{proof}
\begin{cor}
The consensus number of $\text{WRN}_{k}$ is $1$, for every $k\ge3$.
\end{cor}

\section{\label{sec:-Solves-Set}$\text{WRN}_{k}$ Solves $\left(k,k-1\right)$-Set
Consensus}

\subsection{Solution in a System of $k$ Processes}

For any $k\ge3$, a $\text{WRN}_{k}$ object can solve the $\left(k,k-1\right)$-set
consensus task for $k$ processes with unique ids taken from $\left\{ 0,...,k-1\right\} $,
using the following algorithm (also described in algorithm \ref{alg:-Set-consensus-using}):
Assume the processes are $P_{0},\dots,P_{k-1}$, and their values
are $v_{0},\dots,v_{k-1}$. Process $P_{i}$ invokes a \texttt{WRN}
with index $i$ and value $v_{i}$. If the output of the operation,
$t$, is $\bot$, $P_{i}$ decides $v_{i}$. Otherwise, it decides
$t$.

\begin{algorithm}
\begin{algorithmic}[1]
\Function{Propose}{$v_i$}
    \Comment {For process $P_i$, $0 \le i < k$}
  \State $t \gets \Call{WRN}{i, v_i}$
    \Comment {$t$ is a local variable.}
  \If{$t \ne \bot$}
    \Return{$t$}
  \Else{}
    \Return{$v_i$}
  \EndIf
\EndFunction
\end{algorithmic}

\caption{\label{alg:-Set-consensus-using}$\left(k-1\right)$-Set consensus
using a $\text{WRN}_{k}$ object.}
\end{algorithm}
\begin{claim}
Algorithm \ref{alg:-Set-consensus-using} is wait free.
\end{claim}
\begin{claim}
\label{claim:The-first-process}The first process to perform \texttt{WRN}
decides its own proposed value.
\end{claim}
\begin{proof}
Since it is the first one to invoke \texttt{WRN}, the output of \texttt{WRN}
is $\bot$, and hence the process decides on its own proposed value.
\end{proof}
\begin{claim}
\label{claim:Let--be}Let $P_{i}$ be the last process to perform
\texttt{WRN}. So $P_{i}$ decides the proposal of $P_{\left(i+1\right)\mod k}$.
\end{claim}
\begin{proof}
Since $P_{i}$ is the last one to invoke \texttt{WRN}, $P_{\left(i+1\right)\mod k}$
has already completed its \texttt{WRN} invocation. Theretofore, $P_{i}$
receives $v_{\left(i+1\right)\mod k}$ as the output from \texttt{WRN}.
Hence, $P_{i}$ decides the value of $P_{\left(i+1\right)\mod k}$.
\end{proof}
\begin{claim}[Validity]
A process $P_{i}$ can decide its proposed value, or the proposed
value of $P_{\left(i+1\right)\mod k}$.
\end{claim}
\begin{claim}
\label{claim:A-process-}A process $P_{i}$ decides its own proposed
value if $P_{\left(i+1\right)\mod k}$ have not invoked \texttt{WRN}
yet.
\end{claim}
\begin{cor}[$\left(k-1\right)$-agreement]
Assume the proposals are pairwise different (there are exactly $k$
different proposals). So at most $k-1$ values can be decided.
\end{cor}
\begin{proof}
Let $P_{i}$ be the first process to invoke \texttt{WRN}, and $P_{j}$
be the last process to invoke \texttt{WRN}. From claim \ref{claim:The-first-process},
$P_{i}$ decides its proposal. From claim \ref{claim:Let--be}, $P_{j}$
decides the proposal of $P_{\left(j+1\right)\mod k}$. From claim
\ref{claim:A-process-}, no process decides the proposal of $P_{j}$.
\end{proof}
\begin{cor}
Algorithm \ref{alg:-Set-consensus-using} solves the $\left(k-1\right)$-set
consensus task for $k$ processes.
\end{cor}
\begin{cor}
$\text{WRN}_{k}$ solves $\left(n',h\right)$-set consensus task for
any $n'/h \leq 3/2$ in a system with $n'$ processes.
\end{cor}
\begin{cor}
$\text{WRN}_{k}$ cannot be implemented from atomic read-write registers.
Hence, $\text{WRN}_{k}$ is stronger than registers.
\end{cor}

\subsection{Solution in a System with $k$ Participating Processes Out of Many}

Assuming that each process has a unique name in $\left\{ 0,\dots,k-1\right\} $
might be a strong limitation in some models. In this section, we assume
we have at most $k$ participating processes, whose names are taken
from $\left\{ 0,\dots,M-1\right\} $, where $M\gg k$.

In \cite{Afek:1999:FWR:301308.301338,Attiya:1998:AWA:277697.277749}
wait-free algorithms have been shown that use registers only to rename
$k$ processes from $\left\{ 0,\dots,M-1\right\} $ to $k$ unique
names in the range $\left\{ 0,\dots,2k-2\right\} $. So we shall relax
our assumption, and assume now we have at most $k$ participating
processes, whose names are in $\left\{ 0,\dots,2k-2\right\} $. Let
us consider the set of functions $\left\{ 0,\dots,2k-2\right\} \to\left\{ 0,\dots,k-1\right\} $,
call it $\mathcal{F}$. So $\left|\mathcal{F}\right|=\left(2k-1\right)^{k}$
is finite, and we can fix an arbitrary ordering of $\mathcal{F}=\left\{ f_{1},\dots,f_{\left(2k-1\right)^{k}}\right\} $.

The $\left(k-1\right)$-set consensus algorithm for $k$ processes
is described in algorithm \ref{alg:-Set-consensus-for}. It uses $\left(2k-1\right)^{k}$
$\text{WRN}_{k}$ objects, $O_{1},\dots,O_{\left(2k-1\right)^{k}}$.
First, the process name is renamed to be $j\in\left\{ 0,\dots,2k-2\right\} $.
Then, for each $\ell\in\left\{ 1,\dots,\left(2k-1\right)^{k}\right\} $
(in this exact order for all processes), the process invokes \texttt{WRN}
operation of $O_{\ell}$ with the index $f_{\ell}\left(j\right)$,
and the proposed value $v_{j}$. If the result of a \texttt{WRN} operation
returns a value different than $\bot$, the process immediately decides
on this returned value, and does not continue to the next iterations.
If the process received $\bot$ from all the \texttt{WRN} operations
on $O_{1},\dots,O_{\left(2k-1\right)^{k}}$, it decides its proposed
value.

\begin{algorithm}
\begin{algorithmic}[1]
\Function{Propose}{$v$}
    \Comment {For process whose name is in $\left\{ 0, \dots, M -1 \right\}$}
  \State $j \gets \Call{Rename}{}$
    \Comment {$j \in \left\{ 0, \dots, 2k - 2 \right\}$}
  \For {$\ell = 1, \dots, \left( 2k - 1 \right)^k$}
    \State $i \gets f_\ell \left( j \right)$
      \Comment {$i \in \left\{ 0, \dots, k - 1 \right\}$ is a local variable.}
    \State $t \gets O_\ell . \Call{WRN}{i, v}$
      \Comment {$t$ is a local variable.}
    \If {$t \ne \bot$}
      \Return {$t$}
    \EndIf
  \EndFor
  \State \Return{$v$}
    \Comment {Reaching here means $t$ was $\bot$ in all iterations}
\EndFunction
\end{algorithmic}

\caption{\label{alg:-Set-consensus-for}$\left(k-1\right)$-Set consensus for
unnamed processes using $\text{WRN}_{k}$ objects.}
\end{algorithm}

The full proof is left for the full paper, below is a sketch of proof.
The first process to perform \texttt{WRN} in each iteration continues
to the next one, and hence the first process to perform \texttt{WRN}
in the last iteration decides on its own value.

We claim that at most $k-1$ different values are decided by $k$
processes that perform \texttt{Propose}. If at most $k-1$ processes
returned, we are done. Assume that all $k$ processes have returned.
Let $\ell$ be the first (smallest index) iteration in which a process
returned, and let $P$ be the last process that returned from this
iteration. Clearly, process $P$ did not return its proposed value,
and no other process returns $P$'s proposed value. Such an iteration
exists, because there is a mapping $f_{\ell^{\prime}}$, where $\ell^{\prime}$
is not the last iteration, i.e., $\ell^{\prime}<(2k-1)^{k}$, that
maps the processes exactly onto $\left\{ 0,\dots,k-1\right\} $, and
some process must have returned in $\ell^{\prime}$ or before.

\section{\label{sec:The--Case}The $k\le2$ Case}

$\text{WRN}_{1}$ is simply a $\text{SWAP}$ object, in which every
call to \texttt{WRN} returns the previous stored value, and swaps
it with the new value. \cite{H91} showed that the consensus number
of $\text{SWAP}$ is $2$.

Algorithm \ref{alg:Solving-consensus-for} solves the consensus task
for two processes using a $\text{WRN}_{2}$ object. The process $P_{i}$
($i\in\left\{ 0,1\right\} $) invokes the \texttt{WRN} operation of
the $\text{WRN}_{2}$ object with the index $i$ and its proposed
value. If the operation returns $\bot$, then $P_{i}$ decides its
proposed value. Otherwise, it decides the returned value.

\begin{algorithm}
\begin{algorithmic}[1]
\Function{Propose}{$v_i$}
    \Comment {For process $P_i$, $i \in \left\{ 0, 1 \right\}$}
  \State $t \gets \Call{WRN}{i, v_i}$
    \Comment {$t$ is a local variable.}
  \If{$t \ne \bot$}
    \Return{$t$}
  \Else{}
    \Return{$v_i$}
  \EndIf
\EndFunction
\end{algorithmic}

\caption{\label{alg:Solving-consensus-for}Solving consensus for two processes
using a $\text{WRN}_{2}$ object.}
\end{algorithm}

It is easy to see that the invocation of \texttt{WRN} by the first
process returns $\bot$, while the second one returns the proposal
of the first process. Hence, the first process to perform \texttt{WRN}
``wins'', and the second one ``loses'', and the agreement criterion
is achieved. It is also clear that validity is preserved; since a
process returns only its proposed value, or the proposal of the other
process.

\bibliographystyle{plain}
\bibliography{wrn}

\begin{thebibliography}{1}

\bibitem{Afek:2016:DOL:2933057.2933116}
Yehuda Afek, Faith Ellen, and Eli Gafni.
\newblock Deterministic objects: Life beyond consensus.
\newblock In {\em Proceedings of the 2016 ACM Symposium on Principles of
  Distributed Computing}, PODC '16, pages 97--106, New York, NY, USA, 2016.
  ACM.

\bibitem{Afek:1999:FWR:301308.301338}
Yehuda Afek and Michael Merritt.
\newblock Fast, wait-free (2k-1)-renaming.
\newblock In {\em Proceedings of the Eighteenth Annual ACM Symposium on
  Principles of Distributed Computing}, PODC '99, pages 105--112, New York, NY,
  USA, 1999. ACM.

\bibitem{Attiya:1998:AWA:277697.277749}
Hagit Attiya and Arie Fouren.
\newblock Adaptive wait-free algorithms for lattice agreement and renaming
  (extended abstract).
\newblock In {\em Proceedings of the Seventeenth Annual ACM Symposium on
  Principles of Distributed Computing}, PODC '98, pages 277--286, New York, NY,
  USA, 1998. ACM.

\bibitem{borowsky1993implication}
Elizabeth Borowsky and Eli Gafni.
\newblock {\em The implication of the Borowsky-Gafni simulation on the
  set-consensus hierarchy}.
\newblock UCLA Computer Science Department, 1993.

\bibitem{Chaudhuri:1990:AHC:93385.93431}
Soma Chaudhuri.
\newblock Agreement is harder than consensus: Set consensus problems in totally
  asynchronous systems.
\newblock In {\em Proceedings of the Ninth Annual ACM Symposium on Principles
  of Distributed Computing}, PODC '90, pages 311--324, New York, NY, USA, 1990.
  ACM.

\bibitem{C93}
Soma Chaudhuri.
\newblock More choices allow more faults: Set consensus problems in totally
  asynchronous systems.
\newblock {\em Information and Computation}, 105(1):132--158, 1993.

\bibitem{FLP}
Michael~J. Fischer, Nancy~A. Lynch, and Mike Paterson.
\newblock Impossibility of distributed consensus with one faulty process.
\newblock {\em J. {ACM}}, 32(2):374--382, 1985.

\bibitem{Herlihy:1991:IRA:113379.113409}
Maurice Herlihy.
\newblock Impossibility results for asynchronous pram (extended abstract).
\newblock In {\em Proceedings of the Third Annual ACM Symposium on Parallel
  Algorithms and Architectures}, SPAA '91, pages 327--336, New York, NY, USA,
  1991. ACM.

\bibitem{H91}
Maurice Herlihy.
\newblock Wait-free synchronization.
\newblock {\em ACM Transactions on Programming Languages and Systems (TOPLAS)},
  13:124--149, January 1991.

\end{thebibliography}

\end{document}